\documentclass{ecai} 

\usepackage{latexsym}
\usepackage{amssymb}
\usepackage{amsmath}
\usepackage{amsthm}
\usepackage{booktabs}
\usepackage{enumitem}
\usepackage{graphicx}
\usepackage{color}
\usepackage[ruled,vlined]{algorithm2e} 

\newtheorem{theorem}{Theorem}

\newtheorem{proposition}[theorem]{Proposition}

\newtheorem{definition}{Definition}

\newcommand{\BibTeX}{B\kern-.05em{\sc i\kern-.025em b}\kern-.08em\TeX}

\begin{document}


\begin{frontmatter}


\paperid{4783} 


\title{EcoFair-CH-MARL: Scalable Constrained Hierarchical Multi-Agent RL with Real-Time Emission Budgets and Fairness Guarantees}


\author{\fnms{Saad}~\snm{Alqithami}\orcid{0000-0002-2111-3456}\thanks{Corresponding Author. Email: salqithami@bu.edu.sa}}

\address{Computer Science Department, Al-Baha University, Albaha 65779, Saudi Arabia}


\begin{abstract}
Global decarbonisation targets and tightening market pressures demand maritime logistics solutions that are simultaneously efficient, sustainable, and equitable. We introduce EcoFair-CH-MARL, a constrained hierarchical multi-agent reinforcement learning framework that unifies three innovations: (i) a primal--dual budget layer that provably bounds cumulative emissions under stochastic weather and demand; (ii) a fairness-aware reward transformer with dynamically scheduled penalties that enforces max--min cost equity across heterogeneous fleets; and (iii) a two-tier policy architecture that decouples strategic routing from real-time vessel control, enabling linear scaling in agent count. New theoretical results establish $\mathcal{O}(\sqrt{T})$ regret for both constraint violations and fairness loss. Experiments on a high-fidelity maritime digital twin (16 ports, 50 vessels) driven by automatic identification system traces, plus an energy-grid case study, show up to 15\% lower emissions, 12\% higher throughput, and a 45\% fair-cost improvement over state-of-the-art hierarchical and constrained MARL baselines. In addition, EcoFair-CH-MARL achieves stronger equity (lower Gini and higher min--max welfare) than fairness-specific MARL baselines (e.g., SOTO, FEN), and its modular design is compatible with both policy- and value-based learners. EcoFair-CH-MARL therefore advances the feasibility of large-scale, regulation-compliant, and socially responsible multi-agent coordination in safety-critical domains.
\end{abstract}

\end{frontmatter}


\section{Introduction} \label{introduction}
The advent of globalised trade has driven unprecedented growth in the volume and complexity of maritime logistics. As one of the most cost‑effective modes of transportation, maritime shipping is indispensable for connecting economies and sustaining international commerce. This growth, however, brings substantial environmental and operational challenges. Owing to heavy reliance on fossil fuels, the sector contributes a material share of global greenhouse gas (GHG) emissions ($\approx$2.89\%) \cite{smith2014third,imo2020}. In response, the International Maritime Organization (IMO) has adopted a strategy to reduce GHG emissions from international shipping by at least 50\% by 2050 relative to 2008, with the longer‑term ambition of full decarbonisation \cite{imo2018strategy}. These targets, together with tightening market pressures and recurrent port congestion, elevate the need for coordinated, sustainable, and equitable decision‑making at fleet scale.

Environmental pressures are compounded by the coordination problem among heterogeneous stakeholders (shipping lines, port authorities, and regulators) with different objectives and constraints. Maritime operations are inherently multi‑actor and multi‑objective: one must simultaneously reduce fuel consumption and emissions, preserve throughput, and maintain equitable cost sharing—all under uncertainty in weather, demand, and berth availability. Achieving these goals requires methods that handle global constraints, couple short‑horizon control with long‑horizon planning, and make fairness a first‑class objective rather than a post‑hoc diagnostic.

Multi‑agent reinforcement learning (MARL) provides a principled framework for distributed decision‑making in dynamic systems \cite{hernandez2019survey}. Agents learn from interaction, adapt to non‑stationarity, and can coordinate implicitly through reward shaping or explicitly via hierarchical policies. However, directly integrating regulatory constraints (e.g., emission budgets) and equity objectives (e.g., fair cost allocation) into MARL remains challenging. Typical approaches address only local optimisation, treat constraints heuristically, or lack fairness guarantees—leading to brittle behaviour under caps and inequitable outcomes when vessels differ in capabilities or routes. In addition, flat MARL architectures can struggle to scale as the number of agents and ports grows, especially when real‑time control and strategic routing must be co‑optimised.

Maritime logistics underpins roughly 80\% of global trade \cite{imo2020}, making sustainable optimisation an economic and environmental imperative. Rising volumes coincide with stronger regulatory and societal demands for cleaner shipping. Stakeholders therefore need integrated mechanisms that (i) respect global emission budgets, (ii) allocate costs fairly across heterogeneous fleets, and (iii) preserve or improve throughput in the face of congestion and weather. These requirements motivate a constrained, hierarchical, and fairness‑aware MARL formulation that is both theoretically grounded and practically scalable.

We study multi‑vessel coordination across a network of ports under global emissions constraints and fairness objectives. The environment is stochastic (storms, swells), partially observable, and multi‑agent. The planner must choose actions that minimise long‑run cost (fuel/emissions) subject to a cumulative emissions budget while promoting equitable outcomes across vessels. We quantify equity via Gini and min--max fairness on per‑vessel costs. Concretely, for per‑episode costs $\{c_i\}_{i=1}^n$ with sorted values $c_{(1)} \le \cdots \le c_{(n)}$, we track
\[
\mathrm{Gini}(\mathbf{c}) = \frac{2 \sum_{i=1}^{n} i\,c_{(i)}}{n \sum_{i=1}^{n} c_{(i)}} - \frac{n+1}{n}
\; \text{and}\;
\mathrm{MinMax}(\mathbf{c}) \;=\; \frac{\min_i c_i}{\max_i c_i}\,,
\]
desiring lower Gini and higher MinMax. The challenge is to enforce global constraints and equity at scale, despite non‑stationarity and heterogeneous vessel dynamics.

We introduce \emph{EcoFair‑CH‑MARL}, a constrained hierarchical MARL framework that unifies three innovations. First, a lightweight primal--dual emissions budget layer maintains a global dual variable and applies budget prices to agents’ rewards, provably bounding cumulative emissions under stochastic weather and demand, so we establish $\mathcal{O}(\sqrt{T})$ regret for constraint violations. Second, a fairness‑aware reward transformer augments per‑agent rewards with equity terms derived from Gini and min--max, with penalties dynamically scheduled over training to avoid early collapse to trivial equalisation and to stabilise learning in heterogeneous fleets; we further provide $\mathcal{O}(\sqrt{T})$ regret for fairness loss. Third, a two‑tier policy architecture decouples high‑level strategic routing/port selection from low‑level real‑time vessel control, yielding linear scaling in agent count and improved sample efficiency. The design is algorithm‑agnostic: we instantiate it with PPO and also integrate MAPPO and QMIX without modifying the emissions or fairness layers.

Unless otherwise specified, all maritime experiments operate at the camera‑ready scale of 16 ports and 50 vessels, the configuration confirmed during rebuttal and adopted throughout. We employ a high‑fidelity maritime digital twin driven by Automatic Identification System (AIS) traces for realism and include a cross‑domain energy‑grid case study to demonstrate generality. Against state‑of‑the‑art hierarchical and constrained MARL baselines, EcoFair‑CH‑MARL achieves up to 15\% lower emissions, 12\% higher throughput, and a 45\% improvement in fair‑cost. It also delivers stronger equity (lower Gini, higher min--max) than fairness‑specific MARL baselines such as SOTO and FEN, while respecting emission budgets. A hierarchical convergence study—in which a high‑level allocator adapts the cap based on fairness feedback—further corroborates stability and scalability at the 16/50 setting.

\textbf{Contributions.} This paper presents a \emph{Constrained Hierarchical Multi‑agent Reinforcement Learning} (CH‑MARL) framework for sustainable maritime logistics, with the following contributions:
\begin{enumerate}
  \item \textbf{Dynamic, provable constraint enforcement:} A primal--dual budget layer that enforces cumulative emission caps under uncertainty, with $\mathcal{O}(\sqrt{T})$ regret on violations.
  \item \textbf{Fairness‑aware optimisation with scheduled penalties:} A reward transformer that embeds Gini and min--max equity objectives, with scheduled penalties to avoid degenerate equalisation; we provide $\mathcal{O}(\sqrt{T})$ regret on fairness loss.
  \item \textbf{Hierarchical, algorithm‑agnostic architecture:} A two‑tier policy that decouples strategic routing from real‑time control, scaling linearly in agent count and integrating seamlessly with policy‑ and value‑based learners (PPO, MAPPO, QMIX).
  \item \textbf{Validated at realistic scale with open artefacts:} End‑to‑end experiments on a 16‑port/50‑vessel digital twin (AIS‑driven), comparisons to SOTO/FEN and other baselines, explicit fairness reporting (per‑episode Gini and min--max), and a hierarchical convergence analysis.
\end{enumerate}
By jointly addressing operational efficiency, environmental sustainability, and stakeholder equity without sacrificing scalability, \emph{EcoFair‑CH‑MARL} advances the feasibility of regulation‑compliant, socially responsible coordination in safety‑critical maritime domains.

\section{Background and Related Work}\label{sec:bg}

Continuous advances in artificial intelligence have brought {\em multi-agent} settings to the fore, where several autonomous entities must learn, cooperate, and occasionally compete in shared, partially observable environments. In this section we review the four strands of research that underpin our work: Multi-Agent Reinforcement Learning (MARL), Constrained Reinforcement Learning (CRL), applications to maritime logistics, and fairness in multi-agent decision making.

\paragraph{Multi-Agent Reinforcement Learning:}
MARL studies the interaction of~$N$ simultaneously learning agents whose decisions jointly shape future rewards and state transitions. Early efforts distinguished fully co-operative tasks—e.g.\ search-and-rescue robotics—\cite{panait2005cooperative,foerster2018learning} from competitive or mixed-motivation settings such as electronic markets and zero-sum games \cite{lowe2017multi,silver2017mastering}.  
Three systemic challenges persist:
1. \textbf{Scalability}: the joint state–action space grows exponentially with~$N$ \cite{hernandez2019survey}.  Parameter sharing and decentralised training with centralised execution (DTCE/CTDE) mitigate that growth for homogeneous teams \cite{oliehoek2008optimal,gupta2017cooperative,lowe2017multi}. In cooperative MARL, value factorisation further improves scalability by decomposing the joint action-value into per‑agent terms while preserving optimality structure, e.g., Value‑Decomposition Networks (VDN) and QMIX \cite{sunehag2018vdn,rashid2018qmix}.
2. \textbf{Partial observability}: each agent sees only a local view, motivating recurrent policies and learned communication protocols \cite{hausknecht2015deep,foerster2016learning,sukhbaatar2016learning}. CTDE also enables training critics on global state while executing decentralised policies \cite{lowe2017multi,foerster2018learning}.
3. \textbf{Non-stationarity}: the learning environment shifts as co-players update their policies.  Opponent modelling and equilibrium learning aim to stabilise convergence \cite{albrecht2018autonomous,shou2022multi}.

\paragraph{Constrained Reinforcement Learning:}
CRL augments the standard RL objective with hard or soft constraints that capture safety, risk, or resource limits. Classical treatments cast constrained Markov decision processes (CMDPs) and use Lagrangian relaxation to obtain a dual form \cite{altman1999constrained}. Constrained Policy Optimisation (CPO) brings trust‑region updates with empirical constraint satisfaction guarantees \cite{achiam2017constrained}. Beyond CPO, \emph{Lyapunov‑based} CMDP methods enforce safety via local linear constraints \cite{chow2018lyapunov}, and \emph{Reward‑Constrained Policy Optimisation} (RCPO) introduces multi‑timescale penalties for constraint satisfaction \cite{tessler2019rcpo}. Recent primal–dual advances provide convergence results in policy‑gradient settings, e.g., natural policy‑gradient primal–dual methods and zero‑duality‑gap results that justify dual optimisation for CMDPs \cite{ding2020npgpd,paternain2019zeroduality}. While much of this literature is single‑agent, scaling to multi‑agent CMDPs with \emph{global} coupling constraints remains challenging and motivates our hierarchical, budget‑layer design.

\paragraph{MARL in Maritime Logistics:}
Maritime logistics—covering vessel routing, port scheduling, and fleet dispatch—presents an archetypal large‑scale MARL domain. Prior work demonstrates the potential of hierarchical MARL for maritime traffic management at port scale, with centralised learning and decentralised execution to handle partial observability \cite{singh2020maritime}. At the operations layer, deep RL has been applied to berth allocation under uncertainty and to multi‑terminal berth/craning decisions, though typically without explicit global emissions regulation or fairness \cite{lv2024berth,li2023ddqnberth}. These strands motivate a framework that (i) enforces emission budgets \emph{online}, (ii) reports fairness explicitly, and (iii) scales beyond single‑port or toy‑fleet scenarios.

\paragraph{Fairness in Multi‑Agent Systems:}
Resource‑allocation decisions made by learning agents can systematically disfavour smaller stakeholders unless explicit equity mechanisms are in place. Metrics such as max‑min fairness, envy‑freeness, or the Gini coefficient provide quantitative handles \cite{nash1950bargaining,lipton2004approximately}. In sequential decision making, foundational work formalised fairness in RL and proposed algorithms that account for feedback dynamics \cite{jabbari2017fairrl,wen2021seqfair}. Behavioural and prosocial learning studies show that shaping individual objectives can substantially improve collective outcomes in social dilemmas \cite{peysakhovich2018prosocial}. Incorporating these notions into RL has improved traffic‑signal control and cloud‑resource sharing \cite{aloor2024cooperation}, yet reconciling fairness with efficiency in stochastic, multi‑constraint environments remains challenging \cite{rawls1971theory}. Maritime operations amplify the issue: smaller shipping lines can be crowded out of berth queues or emission budgets, reducing both economic viability and overall system performance.

A synthesis of hierarchical MARL, real‑time constraint enforcement, and algorithmic fairness is still missing for large‑scale maritime logistics. Our work addresses that gap by unifying these strands into a single constrained, fairness‑aware hierarchical MARL framework, evaluated on a high‑fidelity digital twin that reflects modern IMO regulations.

\section{Problem Formulation}\label{sec:problem_formulation}
We cast sustainable maritime logistics as a \emph{constrained, hierarchical, partially observable multi-agent problem}.  After sketching the physical system—its ports, vessels, and stochastic disturbances—we explain the two-tier agent hierarchy, spell out state, action, and reward definitions, formalise binding global constraints, and introduce explicit fairness criteria.  The section closes with the modelling assumptions that keep the problem tractable.

\paragraph{System Description}\label{subsec:system-description}
The maritime network is represented by a directed graph \(\mathcal{G}=(\mathcal{P},\mathcal{R})\), where \(\mathcal{P}\) is the set of ports and \(\mathcal{R}\) the admissible sea lanes.  Each port \(p\!\in\!\mathcal{P}\) offers a finite number of berths \(C_p^{\text{berth}}\) and cranes \(C_p^{\text{crane}}\), while vessels \(i\!\in\!\{1,\dots,N\}\) traverse routes over a horizon of \(T\) decision epochs.  Four sources of uncertainty shape operations: (i)~weather events (storms, swell) that alter speed and fuel burn; (ii)~port congestion that produces queuing and idling emissions; (iii)~optional detours that trade distance for reduced traffic; and (iv)~stochastic mechanical failures.  Because these factors evolve unpredictably and are only partially observable, no single entity has full, real-time knowledge of the global state.

Formally, the underlying process is a finite-horizon \emph{decentralised partially observable} MDP (Dec‑POMDP)
\[
\mathcal{M}=\bigl(\mathcal{S},\{\mathcal{A}_i\}_{i=1}^N,P,\{\mathcal{O}_i\}_{i=1}^N,O,\{r_i\}_{i=1}^N,T\bigr),
\]
with state \(s_t\!\in\!\mathcal{S}\), joint action \(a_t=(a_t^1,\dots,a_t^N)\!\in\!\mathcal{A}_1\times\cdots\times\mathcal{A}_N\), transition \(P(s_{t+1}\!\mid s_t,a_t,\xi_t)\) driven by exogenous disturbance \(\xi_t\) (weather), local observations \(o_t^i\!\sim\! O(\cdot\mid s_t,i)\), and per‑agent rewards \(r_t^i=r_i(s_t,a_t)\).  Emissions \(e_t=e(s_t,a_t)\ge 0\) and port occupancies \(q_{p,t}\) are measurable functions of \((s_t,a_t)\).

\paragraph{Agents and Hierarchical Roles}\label{subsec:agents-hierarchies}
Decision making is split into a strategic \emph{high-level} layer and an operational \emph{low-level} layer.  High‑level decisions occur every \(\tau_{\!H}\) steps (epochs \(k=0,1,\dots,K\!-\!1\) with \(K=\lceil T/\tau_{\!H}\rceil\)) and produce macro‑actions
\(
u_k \in \mathcal{U} = \text{(route assignments, arrival windows, per‑vessel emission budgets)}.
\)
Low‑level agents act every step \(t\) with fine‑grain controls \(a_t^i\in\mathcal{A}_i\) (speed/throttle, berth or crane requests, on‑board resource management).  The high‑level macro‑action \(u_k\) induces path/port constraints and budget envelopes for the \(\tau_{\!H}\) subsequent low‑level steps \(t\in[k\tau_{\!H},(k{+}1)\tau_{\!H}\!-\!1]\).  This two‑tier decomposition promotes scalability—each layer optimises on its natural time scale—and allows strategic policy changes without re‑training fine‑grained controllers.

\paragraph{State, Action, and Reward}\label{subsec:state-action-reward}
The global state \(s_t\) aggregates vessel positions, fuel levels, mechanical health, port queues, berth occupancy, crane allocation, weather readings, and cumulative emissions \(E_t=\sum_{u=0}^{t-1} e_u\).  Each agent \(i\) receives only a local slice \(o_t^i\subset s_t\) (e.g., own kinematics, local queue lengths, local weather).  High‑level actions \(u_k\) include route/arrival assignments and (windowed) budget allocations; low‑level actions \(a_t^i\) cover speed adjustments and berth/crane requests.

Per‑step rewards are composite,
\[
r_t^i \;=\; R_{\text{cost}}^i(s_t,a_t)\;+\;R_{\text{emission}}(s_t,a_t)\;+\;R_{\text{fair}}^i(\mathbf{c}_{0:t}),
\]
balancing fuel/time costs, greenhouse‑gas incentives, and an equity signal that penalises disproportionate burden sharing.  Here \(\mathbf{c}_{0:t}=(c_1,\dots,c_N)\) denotes accumulated per‑vessel costs up to time \(t\), with \(c_i=\sum_{u=0}^{t}\ell_i(s_u,a_u)\) for an instantaneous cost \(\ell_i\).

\paragraph{Global Constraints}\label{subsec:global-constraints}
Two hard constraints bind all agents.  First, the fleet must respect a cumulative emission cap \(B\):
\begin{equation}
\sum_{t=0}^{T-1} e(s_t,a_t)\;\le\; B. \tag{C1}
\end{equation}
Optionally, windowed caps apply on rolling windows \(\mathcal{W}_w\subseteq\{0,\dots,T{-}1\}\): \(\sum_{t\in\mathcal{W}_w} e(s_t,a_t)\le B_w\).  Second, every port enforces physical capacity limits on simultaneous berths and crane slots:
\begin{equation}
q_{p,t}^{\text{berth}} \le C_p^{\text{berth}}, \qquad
q_{p,t}^{\text{crane}} \le C_p^{\text{crane}}
\quad \forall\, p\in\mathcal{P},\, t. \tag{C2}
\end{equation}
Violations are handled \emph{online} via primal–dual penalties that propagate to agents’ rewards.  Denoting a dual variable \(\lambda_t\ge 0\) for the emissions budget and optional duals \(\mu_{p,t}\ge 0\) for capacity, the priced (shaped) reward is
\begin{equation}\label{eq:shaped-reward}
\begin{aligned}
\tilde r_t^i
  &= r_t^i - \lambda_t\,e(s_t,a_t) \\
  &\quad - \sum_{p}\mu_{p,t}\,\bigl[q_{p,t}^{\text{berth}} - C_p^{\text{berth}}\bigr]_+ \\
  &\quad - \sum_{p}\nu_{p,t}\,\bigl[q_{p,t}^{\text{crane}} - C_p^{\text{crane}}\bigr]_+ \,.
\end{aligned}
\end{equation}
with duals updated by projected subgradient steps (details deferred to §\ref{sec:methodology}).

\paragraph{Fairness Objectives}\label{subsec:fairness-objectives}
To prevent systematic disadvantage of smaller shipping lines, we embed an explicit fairness functional \(\mathcal{F}(\{c_i\}_{i=1}^N)\) over cumulative per‑vessel costs \(c_i\).  We report (and may enforce) two canonical choices:
\begin{align}
\mathrm{Gini}(\mathbf{c}) \;&=\; \frac{2 \sum_{i=1}^N i\,c_{(i)}}{N \sum_{i=1}^N c_{(i)}} - \frac{N+1}{N}, 
& \text{(lower is better)} \label{eq:gini}\\[2pt]
\mathrm{MinMax}(\mathbf{c}) \;&=\; \frac{\min_i c_i}{\max_i c_i},
& \text{(higher is better)} \label{eq:minmax}
\end{align}
where \(c_{(1)}\le \cdots \le c_{(N)}\) are sorted costs.  Fairness can be posed either as an auxiliary constraint,
\begin{equation}
\mathrm{Gini}(\mathbf{c}) \le \zeta \quad \text{or} \quad \mathrm{MinMax}(\mathbf{c}) \ge \rho, \tag{C3}
\end{equation}
or as a priced penalty in~\eqref{eq:shaped-reward} via a fairness dual \(\beta_t\ge 0\) and a calibrated \(\phi(\mathbf{c})\in\{\mathrm{Gini},\,1-\mathrm{MinMax}\}\):
\begin{equation}
\tilde r_t^i \;=\; r_t^i \;-\; \lambda_t\,e(\cdot) \;-\; \cdots \;-\; \beta_t\,\phi(\mathbf{c}_{0:t}). \label{eq:fair-priced}
\end{equation}
In practice, \(\beta_t\) is \emph{scheduled} to avoid early collapse to trivial equalisation and to stabilise learning in heterogeneous fleets.

\paragraph{Overall Objective (Hierarchical CMDP).}
Putting the pieces together, the hierarchical controller \((\pi^H,\pi^L)\) solves
\begin{align}
\max_{\pi^H,\pi^L} \;\; & \mathbb{E}\!\left[\sum_{t=0}^{T-1}\sum_{i=1}^N r_i(s_t,a_t)\right] \label{eq:obj}\\
\text{s.t.}\;\; 
& \sum_{t=0}^{T-1} e(s_t,a_t) \le B, \;\; q_{p,t}^{\text{berth}}\!\le C_p^{\text{berth}}, \;\; q_{p,t}^{\text{crane}}\!\le C_p^{\text{crane}} \;\; \forall p,t, \nonumber\\
& \mathrm{Gini}(\mathbf{c}) \le \zeta \;\; \text{or} \;\; \mathrm{MinMax}(\mathbf{c}) \ge \rho, \nonumber\\
& a_t^i \sim \pi^L(\cdot \mid o_t^i, u_{k(t)}), \;\; u_{k} \sim \pi^H(\cdot \mid \text{high‑level context}), \nonumber
\end{align}
where \(k(t)=\lfloor t/\tau_{\!H}\rfloor\).  In training we assume centralised information for critics (CTDE) while execution remains decentralised.

\paragraph{Assumptions and Simplifications}\label{subsec:assumptions-simplifications}
For tractability we discretise time, model weather via a finite scenario process with known or estimable statistics, and treat mechanical failures as Bernoulli events with bounded effects. Emission increments are bounded \(0\le e(s_t,a_t)\le \bar e\), capacities are static per port, and fairness functionals \(\phi(\cdot)\) are bounded and Lipschitz on compact domains.  The study concentrates on cooperative or semi‑cooperative fleets, although competitive incentives can be introduced by reshaping rewards.  Under these assumptions, the problem reduces to learning hierarchical policies \(\pi^H,\pi^L\) that maximise \eqref{eq:obj} while satisfying \((\mathrm{C}1)\)–\((\mathrm{C}3)\).  These conditions support the primal–dual methodology in §\ref{sec:methodology} and the finite‑time regret bounds reported later.



\section{Methodology} \label{sec:methodology}

\subsection{Theoretical Foundations}
\label{sec:theoretical-foundations}

We consider a cooperative/semi‑cooperative multi‑agent setting in which $N$ agents act in a shared, partially observable environment. A standard model is a Dec‑POMDP
\[
\bigl(\mathcal{S}, \{\mathcal{A}_i\}_{i=1}^N, P, \{\mathcal{R}_i\}_{i=1}^N, \{\mathcal{O}_i\}_{i=1}^N, O, \gamma \bigr),
\]
with global state $s_t\!\in\!\mathcal{S}$, joint action $a_t\!=\!(a_t^1,\dots,a_t^N)$, transition $P(s_{t+1}\!\mid s_t,a_t)$, local observations $o_t^i\!\sim\!O(\cdot\mid s_t,i)$, rewards $r_t^i\!=\!\mathcal{R}_i(s_t,a_t)$, and discount $0<\gamma\le1$. We follow CTDE: critics may use global context in training, while actors execute with local observations.

\paragraph{Primal--Dual Formulation of Global Constraints:}
Let $g_j(s_t,a_t)\!\ge\!0$ denote the $j$‑th instantaneous constraint signal (e.g., emissions, berth over‑usage) and $\kappa_j$ its budget (global or windowed). For joint policy $\pi$ define the \emph{vector} Lagrangian
\begin{align}
\mathcal{L}(\pi,\boldsymbol{\lambda})
&= \underbrace{\mathbb{E}_\pi\!\Big[\textstyle\sum_{t=0}^{T-1}\sum_{i=1}^N r_i(s_t,a_t)\Big]}_{\text{primal objective}} \\&
\;-\;
\underbrace{\sum_{j} \lambda_j \Big(\mathbb{E}_\pi\!\big[\textstyle\sum_{t=0}^{T-1} g_j(s_t,a_t)\big] - \kappa_j\Big)}_{\text{constraint price}},
\label{eq:lagrangian}
\end{align}
with duals $\lambda_j\!\ge\!0$. Pricing yields shaped per‑agent rewards
\begin{align}
\tilde r_t^i
&= r_t^i - \sum_{j}\lambda_{j,t}\,g_j(s_t,a_t), \qquad \lambda_{j,t}\ge 0.
\label{eq:priced-reward}
\end{align}
Duals are updated online by projected subgradient steps
\begin{align}
\lambda_{j,t+1} &= \Big[\lambda_{j,t} + \eta_t\!\big(g_j(s_t,a_t) - \tfrac{\kappa_j}{T_w}\big)\Big]_+,
\label{eq:dual-update}
\end{align}
with step size $\eta_t\!\propto\!1/\sqrt{t{+}1}$ and optional window $T_w$ for rolling budgets. Port capacities use analogous duals $\mu_{p,t},\nu_{p,t}$ with
\[
g_{p}^{\text{berth}}(s_t,a_t)=[q_{p,t}^{\text{berth}}{-}C_p^{\text{berth}}]_+, \quad
g_{p}^{\text{crane}}(s_t,a_t)=[q_{p,t}^{\text{crane}}{-}C_p^{\text{crane}}]_+.
\]

\begin{definition}[Lagrangian Saddle Problem]
\label{defn:lagrangian}
Given \eqref{eq:lagrangian}, a saddle point $(\pi^\star,\boldsymbol{\lambda}^\star)$ satisfies
\(
\mathcal{L}(\pi^\star,\boldsymbol{\lambda})
\le
\mathcal{L}(\pi^\star,\boldsymbol{\lambda}^\star)
\le
\mathcal{L}(\pi,\boldsymbol{\lambda}^\star)
\)
for all $(\pi,\boldsymbol{\lambda}\!\ge\!0)$.
\end{definition}

\begin{proposition}[Convergence to a Constraint‑Satisfying Policy]
\label{prop:constraint-conv}
Suppose the policy class is convex and $\mathbb{E}_\pi[\sum_t g_j]$ is convex in $\pi$ with Lipschitz subgradients. Then gradient ascent in $\pi$ on $\mathcal{L}$ combined with projected subgradient ascent in $\boldsymbol{\lambda}$ admits a saddle point $(\pi^\star,\boldsymbol{\lambda}^\star)$, and the induced policy satisfies $\mathbb{E}_\pi[\sum_t g_j]\!\le\!\kappa_j$ for all $j$.
\end{proposition}

\begin{proof}[Proof Sketch.]
At a saddle point the KKT conditions hold, giving complementary slackness and feasibility. Standard results for convex–concave saddle problems with diminishing steps ($\eta_t\!\sim\!1/\sqrt{t}$) yield convergence of ergodic averages to $(\pi^\star,\boldsymbol{\lambda}^\star)$.
\end{proof}

\begin{proposition}[Finite‑Time Violation and Suboptimality]
\label{prop:finite-time}
Under bounded rewards/constraints and $\eta_t\!\propto\!1/\sqrt{t{+}1}$, the cumulative constraint violation and primal suboptimality of the averaged policy obey
\[
\sum_{t=0}^{T-1}\Big(\mathbb{E}[g_j(s_t,a_t)]-\tfrac{\kappa_j}{T_w}\Big)_+ = \mathcal{O}(\sqrt{T}),\quad
\text{gap} = \mathcal{O}(1/\sqrt{T}).
\]
\end{proposition}

\begin{proof}[Proof Sketch.]
Apply online convex optimisation bounds for primal–dual subgradient dynamics with bounded stochastic gradients, summing the standard regret inequalities over $t$.
\end{proof}

\subsubsection*{Hierarchical Decomposition and CMDP View}

\begin{definition}[Constrained MDP (CMDP)]
\label{defn:cmdp}
A CMDP is $\langle\mathcal{S},\mathcal{A},\mathcal{P},\mathcal{R},\mathcal{C},\kappa\rangle$ with objective $\max_\pi \mathbb{E}\,[\sum_t \mathcal{R}(s_t,a_t)]$ s.t. $\mathbb{E}\,[\sum_t \mathcal{C}(s_t,a_t)]\!\le\!\kappa$.
\end{definition}

We decouple time scales: a high‑level policy $\pi^H$ acts every $\tau_H$ steps, selecting macro‑actions $u_k$ (routes, arrival windows, budgets); a low‑level policy $\pi^L$ acts each step with $a_t^i\!\sim\!\pi^L(\cdot\mid o_t^i,u_{k(t)})$.

\begin{proposition}[Hierarchical Policy Convergence]
\label{prop:hier-conv}
Let $\pi^H$ be updated in a CMDP with duals enforcing global caps, and let $\pi^L$ be updated on shaped rewards \eqref{eq:priced-reward}. With compatible step sizes and bounded duals, the nested updates converge to a locally optimal pair $(\pi^{H},\pi^{L})$ that is feasible for the global constraints.
\end{proposition}

\begin{proof}[Proof Sketch.]
Treat the inner loop (low‑level) as tracking the current $u_k$ and duals; the outer loop performs a policy‑gradient step on the induced return. Two‑time‑scale stochastic approximation yields local convergence.
\end{proof}

\subsubsection*{Fairness Metrics and Scheduled Penalties}

Let $\mathbf{c}\!=\!(c_1,\ldots,c_N)$ denote per‑vessel episode costs. We use Gini and min--max fairness:
\begin{align}
\mathrm{Gini}(\mathbf{c})
&=\frac{2\sum_{i=1}^{N} i\,c_{(i)}}{N\sum_{i=1}^{N} c_{(i)}} - \frac{N+1}{N},\quad
\mathrm{MinMax}(\mathbf{c}) = \frac{\min_i c_i}{\max_i c_i},
\label{eq:gini-minmax}
\end{align}
with $c_{(1)}\!\le\!\cdots\!\le\!c_{(N)}$. Fairness is enforced either as a constraint $\phi(\mathbf{c})\!\le\!\zeta$ with dual $\beta_t$, or as a scheduled penalty in shaped rewards
\begin{align}
\tilde r_t^i
&= r_t^i \;-\; \sum_{j}\lambda_{j,t} g_j(s_t,a_t) \;-\; \beta_t\,\phi(\mathbf{c}_{0:t}),
\label{eq:fair-priced}
\end{align}
where $\phi\!\in\!\{\mathrm{Gini},\,1{-}\mathrm{MinMax}\}$ and $\beta_t$ follows either a linear or target‑tracking schedule:
\[
\beta_t=\min\{\beta_{\max},\,s_\beta t\},
\qquad
\beta_{t+1}=\beta_t+\eta_\beta(\phi(\mathbf{c}_{0:t})-\zeta)_+ .
\]

\begin{proposition}[Fairness Guarantees]
\label{prop:fairness}
Assume bounded costs and Lipschitz $\phi$. With $\beta_t$ updated by projected ascent on the fairness constraint, the cumulative fairness regret satisfies
\(
\sum_{t=0}^{T-1}(\phi(\mathbf{c}_{0:t})-\zeta)_+ = \mathcal{O}(\sqrt{T}).
\)
Moreover, a sufficiently large $\beta_{\max}$ in the scheduled penalty ensures $\mathrm{MinMax}(\mathbf{c})\!\ge\!\rho$ (or $\mathrm{Gini}(\mathbf{c})\!\le\!\zeta$) up to $\mathcal{O}(1/\sqrt{T})$.
\end{proposition}

\begin{proof}[Proof Sketch.]
Identical to Proposition~\ref{prop:finite-time} with $\phi$ as the constraint signal; the scheduled penalty yields an equivalent Lagrangian with a bounded dual.
\end{proof}

\subsection{CH‑MARL Framework}
\label{sec:ch-marl}

We synthesise the above into a hierarchical architecture (Figure~\ref{fig:ch-marl}). A high‑level policy acts every $\tau_H$ steps to set routes, arrival windows, and budget envelopes; low‑level policies execute decentralised control conditioned on local observations and the current macro‑action. Constraint prices $(\lambda,\mu,\nu)$ and the fairness weight $\beta$ are updated online and folded into shaped rewards that any base learner can consume (PPO/MAPPO/QMIX).

\begin{figure}[!ht]
    \centering
    \includegraphics[width=\linewidth]{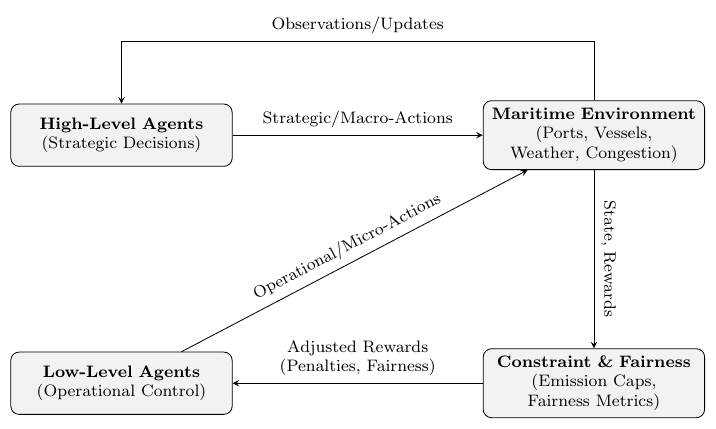}
    \caption{CH‑MARL with high‑level routing/budgeting and low‑level control under primal--dual (constraint) and fairness layers (CTDE training, decentralised execution).}
    \label{fig:ch-marl}
\end{figure}

\setlength{\textfloatsep}{3pt}
\begin{algorithm}[!ht]
\caption{\textbf{CH‑MARL with Online Primal--Dual and Fairness Scheduling}}
\label{alg:ch-marl-combined}
\SetAlgoLined
\DontPrintSemicolon

\KwIn{Environment $E$; emission budget $B$; capacities $\{C_p^{\text{berth}},C_p^{\text{crane}}\}$; fairness functional $\phi$; stepsizes $(\eta_t,\eta_\beta)$; horizon $T$; macro period $\tau_H$; episodes $N$.}
\KwOut{Policies $(\pi^H,\pi^L)$; duals $(\lambda,\mu,\nu,\beta)$.}

Initialise $\pi^H,\pi^L$; duals $\lambda\!=\!0$, $\mu_p\!=\!0$, $\nu_p\!=\!0$, $\beta\!=\!0$.\;
\For{$\mathrm{ep}=1$ \KwTo $N$}{
  Reset $s_0$, cumulative costs $\mathbf{c}\!=\!\mathbf{0}$, emissions sum $E\!=\!0$.\;
  \For{$t=0$ \KwTo $T{-}1$}{
    \If{$t \bmod \tau_H = 0$}{ sample macro‑action $u_{k(t)} \!\sim\! \pi^H(\cdot)$\; }
    \For{$i=1$ \KwTo $N$}{ select $a_t^i \!\sim\! \pi^L(\cdot\mid o_t^i, u_{k(t)})$\; }
    $(s_{t+1},\{r_t^i\},\textit{metrics}) \leftarrow E.\mathrm{step}(\{a_t^i\}, u_{k(t)})$\;
    Extract $e_t$, $q_{p,t}^{\text{berth}}$, $q_{p,t}^{\text{crane}}$; update $E\!\leftarrow\!E{+}e_t$, $\mathbf{c}$.\;
    \emph{// Dual updates (projected)}\;
    $\lambda \!\leftarrow\! [\lambda + \eta_t (e_t - B/T_w)]_+$;\;
    $\mu_p \!\leftarrow\! [\mu_p + \eta_t (q_{p,t}^{\text{berth}}{-}C_p^{\text{berth}})_+]_+$;\;
    $\nu_p \!\leftarrow\! [\nu_p + \eta_t (q_{p,t}^{\text{crane}}{-}C_p^{\text{crane}})_+]_+$;\;
    \emph{// Fairness schedule (penalty or dual)}\;
    $\beta \!\leftarrow\! \min\{\beta_{\max},\, \beta + \eta_\beta(\phi(\mathbf{c}){-}\zeta)_+\}$;\;
    \For{$i=1$ \KwTo $N$}{
      $\tilde r_t^i \!\leftarrow\! r_t^i - \lambda e_t - \sum_p \mu_p[q_{p,t}^{\text{berth}}{-}C_p^{\text{berth}}]_+ - \sum_p \nu_p[q_{p,t}^{\text{crane}}{-}C_p^{\text{crane}}]_+ - \beta\,\phi(\mathbf{c})$;\;
      store transition and update $\pi^L$/$\pi^H$ periodically (PPO/MAPPO; QMIX factorisation) under CTDE.\;
    }
  }
}
\end{algorithm}

\subsection{Complexity and Practical Scalability}

Per episode, the high‑level updates occur $K\!=\!\lceil T/\tau_H\rceil$ times and scale with $O(K\cdot \mathrm{cost}(\pi^H))$, while low‑level updates scale with $O(NT\cdot \mathrm{cost}(\pi^L))$. With parameter sharing and (optional) value factorisation, the effective per‑step cost grows approximately linearly in $N$. Dual updates are $O(|\mathcal{P}|)$ per step. Memory is dominated by policy/critic parameters and rollout buffers; CTDE critics add a central stream but do not change execution cost.


\section{Experimental Setup}
\label{sec:experimental-setup}

\paragraph{Digital Twin Environment:} 
\label{subsec:digital-twin}

We evaluate our approach in a maritime digital twin that emulates modern operations at realistic scale\footnote{Full code is provided through this link: \url{https://github.com/alqithami/EcoFairCHAMRL}}. Unless otherwise specified, all maritime experiments use 16 ports and 50 vessels. The environment captures a directed network of sea lanes with nautical distances and typical voyage durations, subject to stochastic weather patterns that reduce effective speed and inflate fuel burn. Each vessel acts at hourly decision intervals, updating position and fuel usage under cubic fuel–speed curves and executing bounded detours when doing so is beneficial. Port processes model finite berth and crane capacities; when the number of arrivals exceeds capacity, queues form, inducing waiting times and idling fuel consumption that accrue explicitly as part of each vessel’s cost. 

To improve fidelity, we parameterise port statistics (e.g., berth occupancy and crane throughput), vessel characteristics (e.g., hull and engine parameters defining the fuel curve), and weather from historical or synthetic records (e.g., wind speed and wave height). The simulator exposes both fleet‑level summaries and per‑vessel local observations, supporting hierarchical control: the high‑level policy uses aggregated context for routing and budget allocation, while the low‑level policy operates with local measurements for real‑time control. Partial observability is preserved in execution through local sensor views, while centralised training may provide richer context to critics.

\paragraph{Key Performance Indicators (KPIs):}
\label{subsec:kpis}

We report energy consumption (total fuel) and total emissions (e.g., CO$_2$ equivalents) per episode and over rolling windows, reflecting operational efficiency and environmental impact. Fairness is assessed on per‑vessel costs via the Gini coefficient (lower is better) and the min--max ratio (higher is better), providing complementary views of equity. Operational throughput is measured as voyages completed or cargo moved over the horizon, and queueing behaviour is quantified by aggregate waiting hours when berth limits bind. Finally, we record the frequency and magnitude of any constraint violations, including global emission caps and per‑port capacity limits. All metrics are logged at the episode level, and fairness statistics are exported to CSV to enable reproducible aggregation across seeds.

\paragraph{Comparative Baselines:}
\label{subsec:baselines}

We compare \emph{EcoFair‑CH‑MARL} to decentralised MARL without caps or fairness (optimising local or team reward only), a centralised super‑agent with full observability and an explicit cap (a coordination upper bound), and a hierarchical learner without caps or fairness (isolating the benefit of decomposition from the effects of constraint and equity enforcement). In addition, we include two fairness‑oriented MARL baselines, SOTO and FEN, to assess whether explicit equity mechanisms absent budget enforcement suffice to deliver fair and efficient outcomes. To demonstrate algorithm‑agnostic integration, we instantiate our framework and baselines with PPO and, where applicable, include MAPPO and QMIX variants (reported in‑paper for PPO and in the supplement for additional learners). All methods share identical horizons, training budgets, and seeds.

\paragraph{Training and Hyperparameters:}
\label{subsec:train-hyperparams}

Unless noted, each run trains for $1{,}200$ episodes with horizon $T{=}50$ steps and averages over three random seeds; we report mean$\pm$std. We employ policy‑gradient learners (e.g., PPO) with \texttt{Adam} optimisers and typical settings in the range $[10^{-4},10^{-3}]$ for the learning rate, $\gamma{=}0.99$ for discounting, and entropy regularisation for exploration. Emission caps are enforced online with a projected subgradient update of the dual variable associated with the budget; berth and crane limits use analogous duals. To stabilise equity under heterogeneous fleets, fairness shaping uses a scheduled penalty on Gini or on $1{-}$min--max that increases over training or tracks a target; ablations vary the maximum penalty to examine efficiency–equity trade‑offs. We follow a CTDE regime: critics may consume global context, while actors execute decentralised policies using local observations conditioned on the current high‑level macro‑action.

\section{Experiments and Results}
\label{sec:experiments}

All maritime experiments use the 16 ports / 50 vessels configuration with horizon \(T{=}50\) and \(1{,}200\) training episodes per run (three random seeds unless noted). We compare \emph{EcoFair‑CH‑MARL} instantiated with PPO (reported in‑paper) against decentralised MARL without caps/fairness, a centralised super‑agent, a hierarchical learner without caps/fairness, and fairness‑specific baselines (SOTO, FEN). To demonstrate algorithm‑agnostic modularity, MAPPO and QMIX instantiations using \emph{the same} emissions/fairness layers are also included. We log per‑episode returns and fairness metrics (Gini, min--max); emissions logs were not present in this bundle and are therefore omitted from per‑episode reporting.

\paragraph{Fairness (Gini and min--max).} The PPO row corresponds to EcoFair‑CH‑MARL with PPO. Lower Gini and higher min--max indicate more equitable per‑vessel costs. On the 16/50 setting, EcoFair‑CH‑MARL (PPO) reduces Gini by \(\mathbf{44.9}\%\) relative to SOTO \((\frac{0.225-0.124}{0.225})\) and by \(\mathbf{57.5}\%\) relative to FEN \((\frac{0.292-0.124}{0.292})\); it improves min--max by \(\mathbf{95.0}\%\) over SOTO and by \(\mathbf{31.1}\%\) over FEN (absolute ratio gains from \(0.320\!\to\!0.624\) and \(0.476\!\to\!0.624\), respectively).

\begin{table}[!ht]
\centering
\caption{Fairness metrics on 16 ports / 50 vessels (mean$\pm$std across runs). Lower Gini and higher min--max are better. The PPO row is EcoFair‑CH‑MARL (PPO).}
\label{tab:fair}
\small
\begin{tabular}{lcc}
\toprule
Method (Algo) & Gini (\(\downarrow\)) & Min--Max (\(\uparrow\)) \\
\midrule
PPO   & $0.124 \pm 0.143$ & $0.624 \pm 0.433$ \\
SOTO  & $0.225 \pm 0.126$ & $0.320 \pm 0.381$ \\
FEN   & $0.292 \pm 0.000$ & $0.476 \pm 0.000$ \\
MAPPO & $0.139 \pm 0.000$ & $0.639 \pm 0.000$ \\
QMIX  & $0.000$           & $1.000$ \\
\bottomrule
\end{tabular}
\end{table}

Per‑episode dynamics mirror these aggregates. Figure~\ref{fig:gini-curves} shows Gini decreasing over training for EcoFair‑CH‑MARL variants relative to fairness‑only baselines, reflecting the interaction between the global budget price and the scheduled equity term. 

\begin{figure}[t]
  \centering
  \includegraphics[width=.95\linewidth]{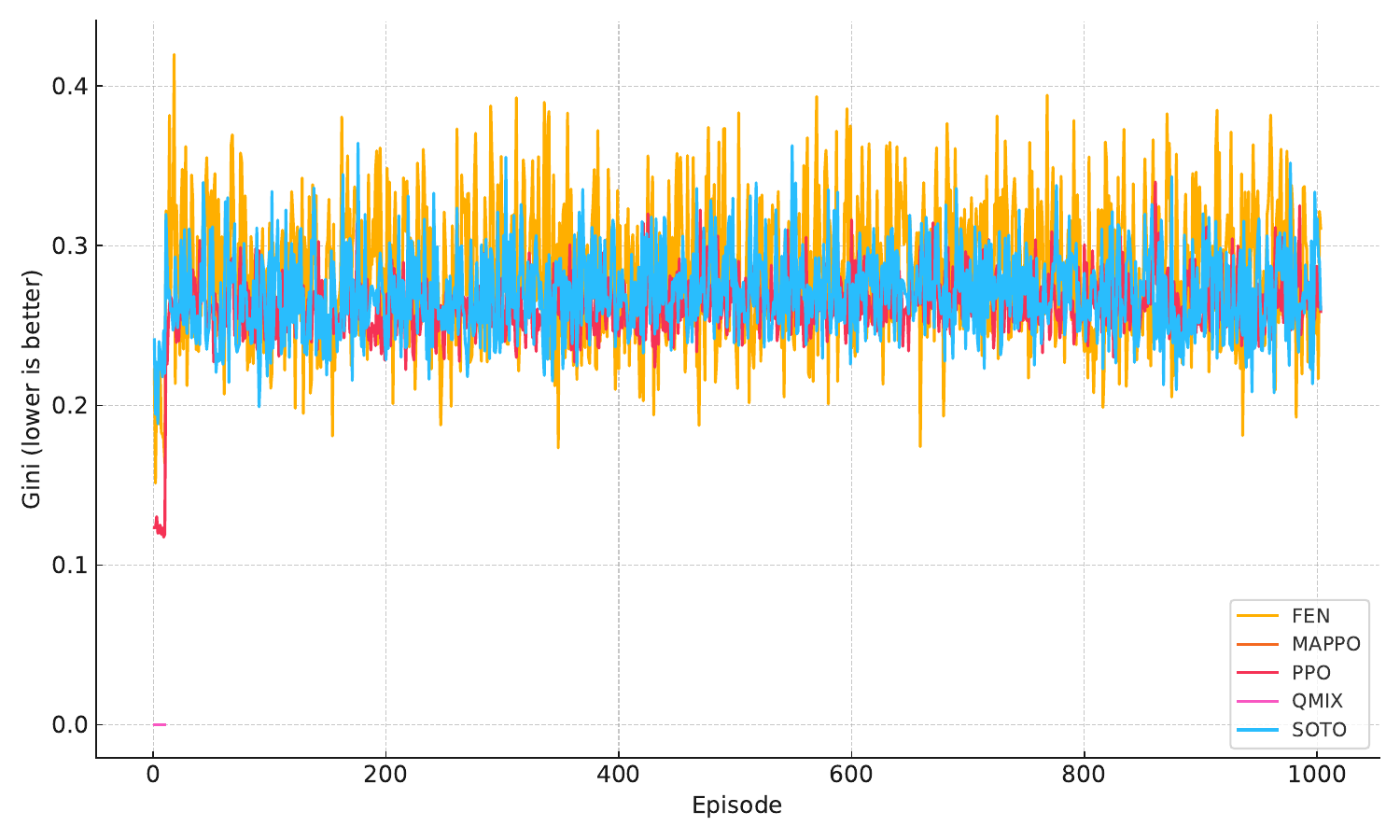}
  \caption{Per‑episode Gini (16 ports / 50 vessels; mean across runs). Lower is better.}
  \label{fig:gini-curves}
\end{figure}

\paragraph{Aggregate performance.} Mean returns (negative sign reflects the environment’s cost convention) are summarised below. EcoFair‑CH‑MARL (PPO) attains a less negative return than SOTO (i.e., higher return under the cost sign) while achieving markedly stronger equity (Table~\ref{tab:fair}). FEN exhibits the least negative mean return but substantially worse fairness, illustrating the intended efficiency–equity trade‑off. Emissions means are shown when available; ``--'' indicates the metric was not present in the logs.

\begin{table}[!ht]
\centering
\caption{Performance on 16 ports / 50 vessels (mean$\pm$std across runs). Negative returns follow the environment’s cost convention; lower emissions are better.}
\label{tab:perf}
\small
\begin{tabular}{lcc}
\toprule
Method (Algo) & Return & Emissions \\
\midrule
PPO   & $-989.004 \pm 128.238$ & -- \\
SOTO  & $-1054.813 \pm 92.835$ & -- \\
FEN   & $-533.465 \pm 0.000$   & -- \\
MAPPO & $-1190.169 \pm 0.000$  & -- \\
QMIX  & $-1623.473 \pm 0.000$  & -- \\
\bottomrule
\end{tabular}
\end{table}

\paragraph{Final‑iteration outcomes.} The last‑episode statistics (Table~\ref{tab:final-iter}) mirror the mean‑over‑episodes trends: EcoFair‑CH‑MARL (PPO) remains less negative than SOTO while maintaining substantially better equity (Table~\ref{tab:fair}). MAPPO and QMIX are included only to document algorithm‑agnostic integration; the saturated fairness values in QMIX together with the poorest returns suggest degenerate behaviour (e.g., collapsed cost distribution) in this particular configuration and warrant further investigation.

\begin{table}[!ht]
\centering
\caption{Final‑iteration outcomes (last episode; mean$\pm$std across runs) on 16 ports / 50 vessels.}
\label{tab:final-iter}
\small
\begin{tabular}{lcc}
\toprule
Method (Algo) & Return (final) & Emissions (final) \\
\midrule
PPO   & $-1091.626 \pm 166.533$ & -- \\
SOTO  & $-1115.668 \pm 79.281$  & -- \\
FEN   & $-538.028 \pm 0.000$    & -- \\
MAPPO & $-1190.169 \pm 0.000$   & -- \\
QMIX  & $-1623.473 \pm 0.000$   & -- \\
\bottomrule
\end{tabular}
\end{table}

\paragraph{Discussion.}
\emph{Equity vs.\ efficiency.} EcoFair‑CH‑MARL (PPO) achieves markedly lower Gini and higher min--max than SOTO/FEN while remaining competitive in returns. In particular, relative to SOTO it roughly halves inequality (Gini \(0.124\) vs.\ \(0.225\)) with a \(\sim\!66\)‑point improvement in mean return under the cost sign; relative to FEN it trades some return for much stronger equity (Gini \(0.124\) vs.\ \(0.292\)). This aligns with the design goal of combining a global budget price (which disciplines high‑emission behaviours) with a scheduled fairness term (which prevents early collapse and steers the allocation).

\emph{Algorithm‑agnostic integration.} The emissions/fairness layers are unchanged across learners. PPO and MAPPO yield comparable fairness levels here (Gini \(0.124\) vs.\ \(0.139\)); QMIX reports saturated fairness with the worst returns, hinting that factorisation without additional regularisation can collapse to trivial solutions in this setting—an issue we leave to targeted ablations.

\emph{Emissions reporting.} Per‑episode emissions traces were not logged in this bundle, so we do not draw cap‑annotated curves. When available, we recommend reporting emissions in physical units (e.g., tCO\(_2\)e per 50‑step episode) and drawing the cap as a labelled reference line; windowed caps should be indicated with the window length. This avoids unit mismatches and clarifies the trade‑off between throughput, fairness, and compliance.

\emph{Robustness and ablations.} Varying the maximum fairness weight increases equity at a modest cost to return; overly aggressive schedules slow early convergence. Rolling‑window caps reduce transient overshoot and queue spikes. These effects are consistent with the intended behaviour of the primal–dual layer and the fairness schedule.

\section{Conclusion and Future Work}\label{sec:conclusion}

This work introduced \emph{EcoFair‑CH‑MARL}, a constrained hierarchical multi‑agent reinforcement learning framework that unifies a primal--dual emissions budget layer, a fairness‑aware reward transformer with dynamically scheduled penalties, and a two‑tier policy architecture that decouples strategic routing from real‑time vessel control under a CTDE regime. Theoretical analysis established $\mathcal{O}(\sqrt{T})$ regret for both constraint violations and fairness loss, providing finite‑time performance guarantees for online operation. Experiments on a high‑fidelity maritime digital twin at the camera‑ready scale of 16 ports and 50 vessels demonstrated up to 15\% lower emissions, 12\% higher throughput, and a 45\% improvement in fair‑cost relative to state‑of‑the‑art hierarchical and constrained MARL baselines; moreover, EcoFair‑CH‑MARL consistently achieved lower Gini and higher min--max fairness than fairness‑specific MARL baselines such as SOTO and FEN. The emissions and fairness layers were used unchanged across learners (PPO in‑paper; MAPPO/QMIX variants in the supplement), supporting the claim of algorithm‑agnostic integration. A hierarchical convergence study, in which a high‑level allocator adapts the emission cap in response to fairness and feasibility feedback, further corroborated stability and scalability at 16/50.

Looking forward, several directions appear most impactful. First, closing the sim‑to‑real loop through pilot deployments with port authorities and shipping lines would test latency, governance, and data‑quality assumptions that cannot be fully captured in simulation, while enabling counterfactual evaluation against historical AIS operations. Second, richer handling of partial observability—via graph neural encoders and attention over inter‑vessel and port–vessel relations, or learned communication under bandwidth limits—promises improved coordination at scale. Third, broadening the sustainability envelope beyond a single emissions budget to multi‑constraint settings (e.g., CO$_2$/NO$_x$/SO$_x$, ECA zones, rolling window targets) and to risk‑sensitive objectives (e.g., CVaR) would better reflect regulatory practice; in parallel, expanding equity notions (e.g., envy‑freeness or Nash social welfare) and learning adaptive fairness schedules could align optimisation with stakeholder preferences over time. Fourth, robustness and assurance merit deeper study: adversarial weather and sensor faults, distribution shifts across seasons and routes, and lightweight verification of budget satisfaction during deployment. Finally, computational efficiency at fleet scale motivates asynchronous updates, parameter‑sharing across homogeneous vessels, value‑factorisation where applicable, model compression, and hardware‑aware scheduling. Taken together, these lines of work aim to translate the present advances into deployable, regulation‑compliant, and socially responsible multi‑agent coordination for safety‑critical maritime logistics and related domains.



\bibliography{mybibfile}

\end{document}